\documentclass[12pt,a4paper]{article}
\usepackage{latexsym, amsfonts, amsthm,amssymb,amscd,amsmath,makeidx}
\usepackage{enumerate}
\usepackage{exscale}
\usepackage{color} 
\usepackage{graphicx}
\usepackage{overpic} 
\usepackage{bm}
\usepackage{bbm}
\usepackage{fancyhdr}
\usepackage{mathrsfs} 
\usepackage[normalem]{ulem}

\definecolor{DarkGreen}{rgb}{0.2,0.6,0.2}

\def\Om{\Omega}\def\om{\omega}

\numberwithin{equation}{section}

\def\Ind#1{{\mathbbmss 1}_{_{\scriptstyle #1}}}

\def\lra{\longrightarrow}

\def\da{\downarrow}

\def\wt{\widetilde}

\def\ignore#1{}

\def\bR{{\mathbb R}}

\def\bZ{\mathbb Z}
\def\bN{\mathbb N}

\def\bP{{\mathbb P}}
\def\bE{{\mathbb E}}

\def\cF{{\mathscr F}}
\def\cG{{\mathscr G}}

\def\cM{{\mathscr M}}

\def\cX{{\mathscr X}}

\def\<{\langle}\def\>{\rangle}

\newtheorem{theorem}{Theorem}[section]
\newtheorem{proposition}[theorem]{Proposition}

\newtheorem{lemma}[theorem]{Lemma}

\theoremstyle{definition}

\newtheorem{remark}[theorem]{Remark}

\normalsize

\usepackage{fancyhdr}
\title{Optimal portfolio liquidation in target zone models \\ and catalytic superprocesses}
\author{ \normalsize Eyal Neuman\\  \normalsize Institute for Advanced Study  \\ \normalsize Hong Kong University of Science and Technology \\
 \normalsize Hong Kong 
 \and \normalsize Alexander Schied\\  \normalsize Department of Mathematics\\
 \normalsize University of Mannheim\\
 \normalsize 68131 Mannheim, Germany}
 
\date{\small  First version: April 21, 2015\\
This version: June 25, 2015
 }

\begin{document}

\maketitle
\vspace{-0.5cm}

\begin{abstract}
We study optimal buying and selling strategies in target zone models. In these models the price is modeled by a diffusion process which is reflected at one or more barriers. Such models arise for example when a currency exchange rate is kept above a certain threshold due to central bank intervention. 
We consider the optimal portfolio liquidation problem for an investor for whom prices are optimal at the barrier and who creates temporary  price impact. This problem will be formulated as the minimization of a cost-risk functional over strategies that only trade when the price process is located at the barrier. 
  We solve the corresponding singular stochastic control problem by means of a scaling limit of critical branching particle systems, which is known as a catalytic superprocess. In this setting the catalyst is a set of points which is given by the barriers of the price process. For the cases in which the unaffected price process is a reflected arithmetic or geometric Brownian motion with drift, we moreover give a detailed financial justification of our cost functional by means of an approximation with discrete-time models.
\end{abstract}

\bigskip


\section{Introduction}
We consider currency exchange rate models in which the exchange rate is allowed to move inside a \emph{target zone} with one or more barriers which are enforced by monetary authorities. Target zone models were pioneered by Krugman in~\cite{Krugman91} and also studied in~\cite{Svensson, Bertola-Caballero92, De-Jong94, Ball-Roma98},  among others. In these models, the price process is modeled by a diffusion process which is reflected at one or more barriers. For instance one can think of a currency exchange rate that is kept above a certain threshold due to central bank intervention, such as it has recently been the case for the  rate of the Swiss Franc against the Euro. This rate has been kept above 1.20 SFR/EUR by the Swiss National Bank. If in a such a situation Swiss Francs are sold, then it is natural to do so when the exchange rate is as close as possible to the barrier of reflection. 

Market impact refers to the empirical fact that the execution of a large order affects the price of the underlying asset. Usually, this affect causes an unfavorable additional execution costs for the trader who is performing the exchange. As a result, a trader who wishes to minimize his trading costs has to split his order into a sequence of smaller orders which are executed over a finite time horizon. One of the well studied market impact models is the Almgren--Chriss model (see e.g.~\cite{AlmgrenChriss1,AlmgrenChriss2, Almgren} and references therein, see also~\cite{GatheralSchied} for a survey paper). In the framework of the Almgren-Chriss model it is assumed that the trading transactions cause a temporary imbalance in the supply and demand of the asset which is creating some temporary price movements. This type of affect on  prices is known as \emph{temporary price impact}. In~\cite{AlmgrenChriss1,AlmgrenChriss2, Almgren, AlmgrenHauptmanLi,Gatheral} among others, this temporary price impact is modeled as a function of the trading speed. 

In this work we consider trading strategies for an investor for whom prices are optimal at the barrier and who creates temporary price impact. In this situation, trading will be constrained to those times at which the price process is located \emph{at} the barrier. Such strategies cannot be absolutely continuous in time, and so  trading speed needs to be defined in a manner different from the Almgren--Chriss model. In Section~\ref{approx section}
 we will argue by means of an approximation with discrete-time models of Almgren--Chriss-type that strategies should instead be absolutely continuous with respect to the local time of the price process at the barrier and that trading speed should be defined as the corresponding Radon--Nikodym derivative. We thus arrive at an appropriate notion of temporary price impact and the resulting transaction costs in our setting.

We formulate the minimization of a cost-risk functional of such strategies as a singular stochastic control problem. 
In addition to the expected costs resulting from temporary price impact, the cost functional includes penalizations for current and final deviations from a target positions. 
In Section~\ref{sec-sol}  we solve this control problem by means of a scaling limit of critical branching particle systems, which is known as a catalytic superprocess.   As in~\cite{SchiedFuel}, the nonlinearity of the penalization terms will correspond to the offspring distribution of particles in the superprocess population but, in contrast to~\cite{SchiedFuel}, branching will only take place for particles that are located \emph{at} the barrier of the process. That is, the corresponding superprocess is a \emph{catalytic superprocesses} with a point catalyst given by the barrier of the price process. Such catalytic superprocesses have been widely studied for the case in which the one-particle motion is standard Brownian motion~\cite{DF94,Del96}; a construction for the case in which the one-particle motion is a general diffusion process was given in~\cite{Dynkin95}. 

 Section~\ref{sec-proofs} contains the proofs of our results.

\section{Statement of Results} \label{sec-res} 

We consider a Markovian model $(\Om,\cF,(\cF_t)_{t\ge0}, (P_z), (S_t)_{t\ge0})$ where, under $P_z$, the process $S=(S_t)_{t\ge0}$ is a diffusion process starting from $z\in\bR$ (or $z\in\bR_+$) and reflecting at a barrier $c\in\bR$. Such reflecting diffusion processes have often be proposed as models for currency exchange rates in a target zone~\cite{Krugman91,Bertola-Caballero92,Ball-Roma98}. Here, a target zone refers to a regime in which the exchange rate of a currency is kept within a certain range of values, either through an international agreement or through central bank intervention.  An example of such a target zone is the recent lower bound of 1.20 EUR/CHF  that was enforced through the Swiss National Bank;  see Figure~\ref{EUR_CHF Fig}.  A realization of reflected geometric Brownian motion with a reflecting barrier at $c=1.2$ is given in Figure~\ref{GBM7.pdf} for comparison.

We denote by $L=(L_t)_{t\ge0}$ the local time of $S$ at the barrier $c$. 
Let $\cX$ denote the class of all progressively measurable control processes $\xi$ for which $\int_0^T|\xi_t|\,dL_t<\infty$ $P_z$-a.s.~for all $T>0$ and $z\in\bR$. For $\xi\in\cX$ and $x_0\in\bR$ we define
$$X^{\xi}_t:=x_0+\int_0^t\xi_s\,dL_s.
$$
In the context of a target zone model, the control $\xi$ will be interpreted as a trading strategy that executes orders at infinitesimal rate $\xi_t\,dL_t$ at those times $t$ at which $S_t=c$. For instance, for a Swiss investor wishing to purchase euros during the period of a lower bound on the EUR/CHF exchange rate, it would have been natural to buy euros only at the lowest possible price, which at this time was equal to the barrier of 1.20 EUR/CHF. The resulting process $X^{\xi}_t=x_0+\int_0^t\xi_s\,dL_s$ therefore describes the inventory of the investor at time $t$.

Our goal is to minimize the functional 
\begin{align}\label{cost functional}
E_z\bigg[\,\int_0^T|\xi_t|^p\,L(dt)+\int_0^T\phi(S_t)|X^\xi_t|^p\,dt+\varrho(S_T)|X^\xi_T|^p\,\bigg],
\end{align}
over $\xi\in\cX$. Here, $p\in[2,\infty)$, $\phi:\bR\to\bR_+$ is bounded and measurable, and $\varrho:\bR\to\bR_+$ is bounded and continuous. 
As we will argue in Section~\ref{approx section} in some detail, 
the term $\int_0^T|\xi_t|^p\,L(dt)$ in~\eqref{cost functional} can  be interpreted as a cost term that arises from the temporary price impact generated by  the strategy $X^\xi$. The expectation of $\int_0^T\phi(S_t)|X^\xi_t|^p\,dt$ can be regarded as a measure for the risk associated with holding the position $X^\xi_t$ at time $t$; see~\cite{AlmgrenSIFIN,Forsythetal,Tseetal}
and the discussion in Section 1.2 of~\cite{SchiedFuel}. Similarly, the expectation of the term $\varrho(S_T)|X^\xi_T|^p$ can be viewed as a penalty for still keeping the position $X^\xi_T$ at the end of the trading horizon. In view of the preceding two terms, it is natural to take $x_0<0$ when buying and $x_0>0$ when selling. A strategy minimizing the functional~\eqref{cost functional} can then be interpreted as a trading strategy that liquidates an optimal share of the initial inventory $x_0$ in a cost-risk-efficient manner over the time interval $[0,T]$. The connections with optimal trading under price impact in target zone models will be further investigated in Section~\ref{approx section}.

\begin{figure}
\begin{center}
\includegraphics[width=11cm]{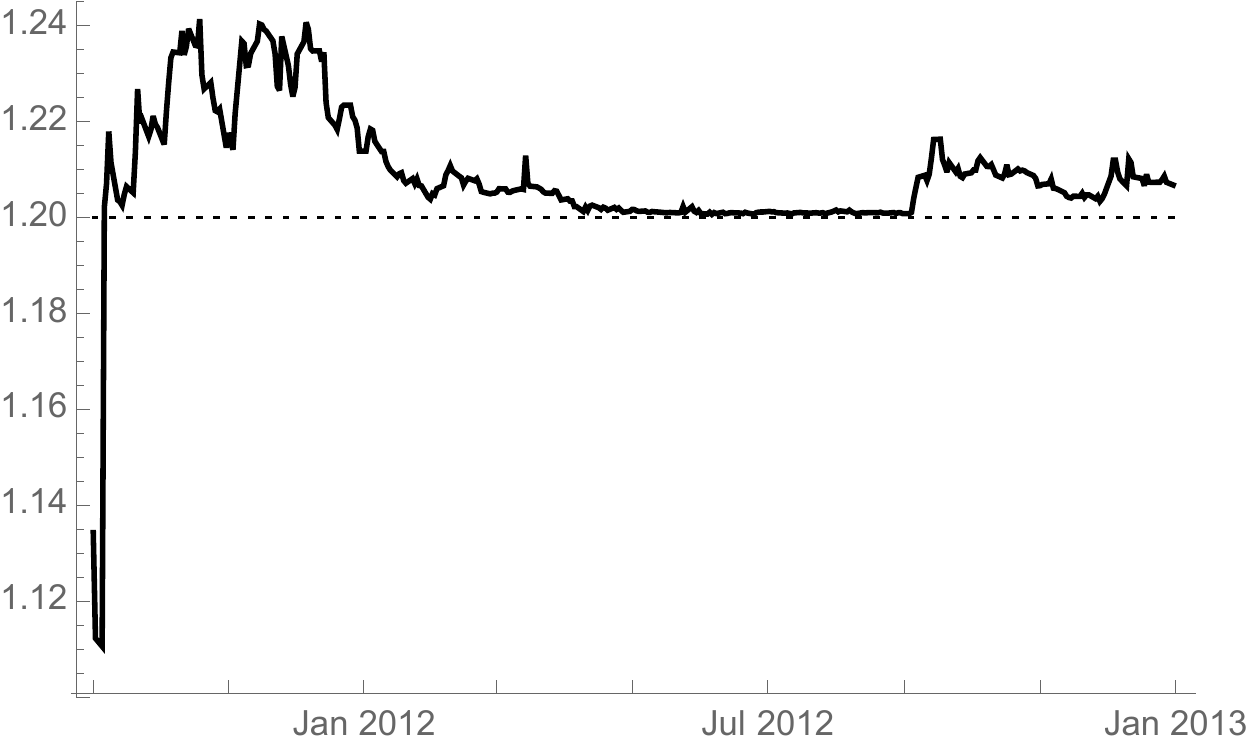}
\caption{Plot of the EUR/CHF exchange rate from  September 1, 2011 through December 31, 2012. On September 6, 2011, the Swiss National Bank announced that it would enforce a minimum exchange rate of 1.20 EUR/CHF. This announcement was  followed by the steep upward spike seen on the left of the plot. }\label{EUR_CHF Fig}
\end{center}
\end{figure}

\begin{figure} \label{GBM7.pdf}
\begin{center}
\includegraphics[width=11cm]{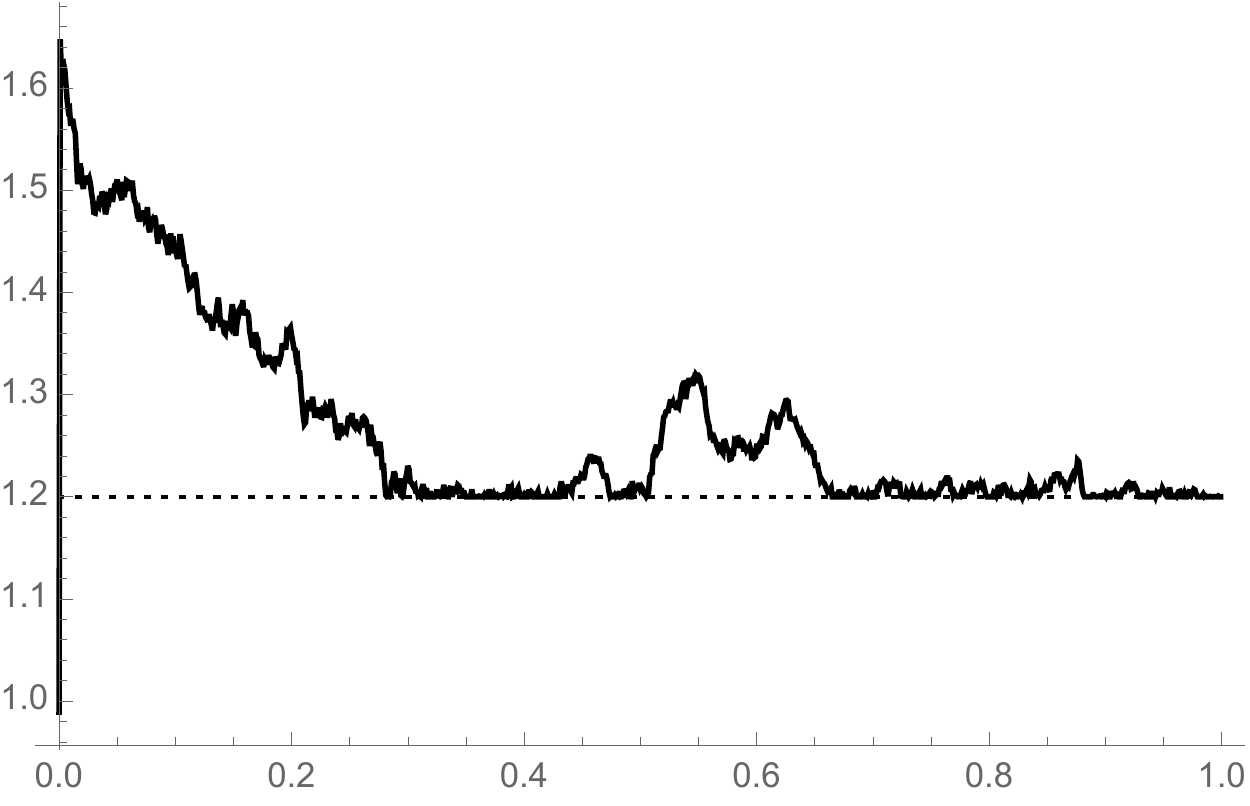}
\caption{Realization of reflected geometric Brownian motion with a reflecting barrier at $c=1.2$.}
\end{center}
\end{figure}

\subsection{Solution of the control problem} \label{sec-sol} 

Our first goal is to solve the problem of minimizing the functional~\eqref{cost functional}. In particular, we will show existence and uniqueness of minimizers and provide a probabilistic representation of the optimal strategy. This representation is based on the following concept of a \emph{catalytic superprocess} with point catalyst $c$, which was introduced by Dawson and Fleischmann~\cite{DF94}.
The following proposition is based on the general existence result of~\cite[Theorem III.4.1]{DynkinBranchingBook} (see also~\cite{Dynkin95}).  By $\cM=\cM(\bR)$ we will denote the space of nonnegative Borel measures on $\bR$, and we will use the shorthand notation  $\<f,\mu\>$ for  $\int f\,d\mu$.

\begin{proposition}\label{superprocess prop}For   $\beta\in(0,1]$, there exists a Markov process $(\wt\Omega,\cG,(\cG_t)_{t\ge0},(\bP_\mu)_{\mu\in\cM},(Y_t)_{t\ge0})$ taking values in $\cM$  such that for each $T>0$, bounded measurable functions $f,g\ge0$, and $\mu\in\cM$,
$$
\bE_\mu\Big[\,\exp\Big(-\int_0^T\<g,Y_t\>\,dt-\<f,Y_T\>\Big)\,\Big]=e^{-\<v(T,\cdot),\mu\>},
$$
where the measurable function $v$ is the unique nonnegative solution of the nonlinear integral equation
\begin{align}\label{logLaplace eqn}
v(t,z)=E_z\Big[\,f(S_t)+\int_0^tg(S_s)\,ds\,\Big]-E_z\Big[\,\frac1\beta\int_0^tv(t-s,S_s)^{1+\beta}\,L(ds)\,\Big].
\end{align}

\end{proposition}

\bigskip

Now we can state our first main result concerning the minimization of the cost functional~\eqref{cost functional}. Recall that $\phi:\bR\to\bR_+$ is bounded and measurable, $\varrho:\bR\to\bR_+$ is bounded and continuous, and $p\in[2,\infty)$. Define $\beta\in(0,1]$ by $\beta:=1/(p-1)$ and consider the corresponding catalytic superprocess from the preceding proposition. We let 
\begin{align}\label{u def eq}
u(t,z):=-\log \bE_{\delta_z}\Big[\,\exp\Big(-\int_0^t\<\phi,Y_s\>\,ds-\<\varrho,Y_t\>\Big)\,\Big]
\end{align}
so that $u$ solves~\eqref{logLaplace eqn}
 for $\varrho$ replacing $f$ and $\phi$ replacing $g$.  \medskip   \\
 We say that $\xi^*$ is the $dL_t\otimes P_z$-a.e.~unique strategy in $\cX$ that is minimizing the cost functional~\eqref{cost functional}, if any two strategies in $\cX$ which minimize~\eqref{cost functional} are equal $dL_t\otimes P_z$-almost everywhere.
 
 \begin{theorem}\label{main thm}With the notation introduced above, let
 \begin{align}\label{optimal strategy formula}
 \xi^*_t:=-x_0\exp\Big(-\int_0^tu(T-s,S_s)^\beta\,dL_s\Big)u(T-t,S_t)^\beta
 \end{align}
 so that
 $$X^{\xi^*}_t=x_0\exp\Big(-\int_0^tu(T-s,S_s)^\beta\,dL_s\Big).
 $$
 Then $\xi^*$ is the $dL_t\otimes P_z$-a.e.~unique strategy in $\cX$ minimizing the cost functional~\eqref{cost functional}. Moreover, the minimal cost is given by 
$$E_z\bigg[\,\int_0^T|\xi^*_t|^p\,L(dt)+\int_0^T\phi(S_t)|X^{\xi^*}_t|^p\,dt+\varrho(S_T)|X^{\xi^*}_T|^p\,\bigg]=|x_0|^pu(T,z).
$$
 \end{theorem}
 

\subsection{Temporary price impact and a discrete-time approximation of the cost functional}\label{approx section}

In this section we explore in greater detail the interpretation of the term $\int_0^T|\xi_t|^p\,L(dt)$ in~\eqref{cost functional} as the cost arising from the temporary price impact created by the strategy $X^\xi$. The notion of temporary price impact originated in the discrete-time framework of~\cite{BertsimasLo,AlmgrenChriss2} and was later conveyed to continuous time by means of a limiting procedure~\cite{Almgren}. Here, we will follow a similar line of reasoning. To this end, we need discrete-time approximations of the asset price process $S$. Natural approximations are available in two cases, namely when $S$ is a reflecting arithmetic Brownian motion with drift, or when $S$ is a reflecting geometric Brownian motion. Arithmetic Brownian motion can be approximated by random walks, geometric Brownian motion by binomial models. For this reason, we will focus here on the two respective cases in which $S$ is a reflecting arithmetic or geometric Brownian motion. 
\bigskip \\
{\bf Case 1: $\bm S$ is a reflecting arithmetic Brownian motion with drift}. That is,
$$S_t=S_0+\sigma B_t+bt+U_t,
$$
where $S_0\in \bR$ is the starting point,  $\sigma>0$ is a volatility parameter, $b\in\bR$ is the drift, $B$ is a standard Brownian motion, and $U$ is the unique continuous, increasing, and adapted process such that  the measure $dU_t$ is supported on $\{t\,|\,S_t=c\}$ and  $S_t\ge c$ (the case in which $c$ is an upper barrier for $S$, i.e., $S_t\le c$, can be handled in the same manner). 
 We now recursively define for 
 $n\in\bN$ fixed the following sequence of stopping times,
 \begin{align}\label{taunkfor ABM}
 \tau_0^{(n)}:=\inf\big\{t\ge0\,\big|\,S_t\in c+2^{-n}\bZ\big\},\qquad \tau_{k}^{(n)}:=\inf\big\{t>\tau_{k-1}^{(n)}\,\big|\,|S_t-S_{\tau_{k-1}^{(n)}}|=2^{-n}\big\}.
 \end{align}
Then we introduce the discretized price process
\begin{align}\label{Sn definition}
S^{(n)}_k:=S_{\tau_k^{(n)}},\qquad k=0,1,\dots
\end{align}
Then $S^{(n)}$ is a reflecting random walk with drift. The (normalized) local time of $S^{(n)}$ in $c$ is usually defined as
\begin{align} \label{bm-dics-lt} 
\ell^{(n)}_k:=2^{-n}\sum_{i=0}^k\sigma\Ind{\{S^{(n)}_i=c\}},
\end{align} 
(see, e.g.,~\cite{LeGallDEA}).
\bigskip\\ 
{\bf Case 2:  $\bm S$ is a reflecting geometric Brownian motion}, i.e., $S$ solves the stochastic differential equation
\begin{align}\label{reflecting GBM}
dS_t= \sigma S_t\,dB_t+b S_t\,dt+dU_t
\end{align}
with starting point $S_0>0$, where $\sigma>0$ is the volatility, $B$ is a standard Brownian motion, $b\in\bR$ is a drift parameter, and  $U$ is the unique continuous, increasing, and adapted process such that the measure $dU_t$ is supported on $\{t\,|\,S_t=c\}$ and $S_t\ge c$ (again, the case $S_t\le c$ is analogous); see, e.g.,~\cite{Skorohod}. In this context, we let
$$\tau_0^{(n)}:=\inf\big\{t\ge0\,\big|\,\log S_t\in \log c+2^{-n}\bZ\big\},\qquad \tau_{k}^{(n)}:=\inf\bigg\{t>\tau_{k-1}^{(n)}\,\Big|\frac{S_t}{S_{\tau_{k-1}^{(n)}}}\in\{d_n,u_n\}\bigg\},
$$
where $d_n=e^{-2^{-n}}$ and $u_n=e^{2^{n}}$. If we now define $S^{(n)}$ as in~\eqref{Sn definition}, then $S^{(n)}$  is a binomial model with parameters $d_n$ and $u_n$, reflecting at $c$. The normalized local time of $S^{(n)}$ in $c$ will be defined as 
$$\ell^{(n)}_k:=2^{-n}\sum_{i=0}^k\sigma c\Ind{\{S^{(n)}_i=c\}}.
$$

The following  convergence result, which is based on the corresponding well-known result for ordinary Brownian motion,  states in that the discrete-time price processes $S^{(n)}$ and the corresponding local times $\ell^{(n)}$ approximate the continuous-time processes $S$ and $L$ after a suitable time change.

\begin{lemma}\label{Skorohod lemma} Under the assumptions of Case 1 or  Case 2, we have that for each $T>0$,  $P_z$-a.s.,
$$
\sup_{t\le T}\big|S^{(n)}_{\lfloor 2^{2n}t\rfloor}-S_{t}\big|\lra0\qquad\text{and}\qquad \sup_{t\le T}\big|\ell^{(n)}_{\lfloor 2^{2n}t\rfloor}-L_{t}\big|\lra0.
$$
\end{lemma}

We now define a discrete-time approximation, $X^{\xi,(n)}$, of the inventory process  $X^{\xi}_t=x_0+\int_0^t\xi_s\,dL_s$ for a fixed control $\xi\in\cX$. For simplicity we will focus on the case in which  $\xi_t(\om)$ is, for $P_z$-a.e. $\om\in\Om$, a continuous function of $t\in\bR_+$. The class of all such controls will be denoted by $\cX_c$. This assumption is justified by the following result.

\begin{proposition}\label{continuous strategies lemma}The function $u(t,z)$ defined in~\eqref{u def eq} is jointly continuous in $t$ and $z$. In particular, the optimal strategy $\xi^*$  from Theorem~\ref{main thm} belongs to $\cX_c$.
\end{proposition}

 Setting $\ell^{(n)}_{-1}:=0$, we now let 
 \begin{align}\label{discrete time orders eq}
 \xi^{(n)}_k:=\xi_{\tau^{(n)}_k}\qquad\text{and}\qquad X^{\xi,(n)}_N:=x_0+\sum_{k=0}^N\xi^{(n)}_k(\ell^{(n)}_k-\ell^{(n)}_{k-1})
 \end{align}
so that $\xi^{(n)}$ is the speed, relative to the local time $\ell^{(n)}$, at which shares are sold or purchased. The random variable $X^{\xi,(n)}_N$ will consequently describe the resulting inventory of the investor at the $N^{\text{th}}$ time step of the discrete-time approximation. 
 Exactly as in the framework of Almgren and Chriss~\cite{AlmgrenChriss2} and Almgren~\cite{Almgren}, each executed order will thus generate temporary price impact, and this price impact is given by a certain odd function $g$ applied to its trading speed. Here we take $g(x)=\text{sign}(x)|x|^{p-1}$ for some $p>1$, which is consistent with~\cite{Almgren,AlmgrenHauptmanLi,Gatheral}. The transaction costs incurred by an order are consequently given by the number of shares executed times the temporary price impact. Since the execution speed of the $k^{\text{th}}$ order in~\eqref{discrete time orders eq} is $\xi^{(n)}_k$ and the number of shares executed by that order is $\nu_k:=\xi^{(n)}_k(\ell^{(n)}_k-\ell^{(n)}_{k-1})$, its  transaction costs are given by
 $$\nu_kg(\xi^{(n)}_k)=|\xi^{(n)}_k|^p(\ell^{(n)}_k-\ell^{(n)}_{k-1}).
 $$ 
It follows that the total transaction costs incurred by the first $N$ orders in~\eqref{discrete time orders eq} are thus equal to
$$\sum_{k=0}^N|\xi^{(n)}_k|^p(\ell^{(n)}_k-\ell^{(n)}_{k-1}).
$$
The following result now provides the desired financial interpretation of the cost minimization problem solved in Theorem~\ref{main thm}. 

\begin{proposition} \label{convergence prop}Under the above assumptions, we have that $P_z$-a.s.~for each $t>0$,
$$X^{\xi,(n)}_{\lfloor 2^{2n}t\rfloor}\lra X^\xi_t\qquad\text{and}\qquad
\sum_{k=0}^{\lfloor 2^{2n}t\rfloor}|\xi^{(n)}_k|^p(\ell^{(n)}_k-\ell^{(n)}_{k-1})\lra\int_0^t|\xi_s|^p\,L(ds).
$$
\end{proposition}

\section{Proofs} \label{sec-proofs}

\begin{proof}[Proof of Proposition~\ref{superprocess prop}] When $p_t(x,y)$ denotes the transition density of the diffusion $S$, we have
\begin{align}\label{local time identity}
E_z[\,L_t\,]=\int_0^tp_s(z,c)\,ds.
\end{align}
Therefore, $L$ is an admissible functional as defined in the beginning of Section 3.3.3 in~\cite{DynkinBranchingBook}, and~\cite[Theorem III.4.1]{DynkinBranchingBook}  yields the existence of the catalytic superprocess.  That we may use a time-homogeneous  formalism follows from Section 1.2.6 in~\cite{DynkinBranchingBook}. Next, by Section 2.4.1 of~\cite{DynkinBranchingBook}, $A(dt):=\varrho(S_T)\delta_T(dt)+\phi(S_t)\,dt$ is a natural additive functional.  From  Section 3.4.5 and  Theorem 3.4.2 in~\cite{DynkinBranchingBook} we get that $v$ satisfies~\eqref{logLaplace eqn}.
\end{proof}

We now prepare for the proof  of Theorem~\ref{main thm} by means of the following two auxiliary lemmas. 

\begin{lemma}  \label{lemma-mon-strat} Suppose $x_0\in\bR$ is given. 
For any  admissible strategy $\xi\in\cX$ for which $X^\xi$ is not monotone or changes sign, there exists another admissible strategy $\eta\in\cX$ that has a strictly lower cost than $\xi$ and for which $X^{\eta}$ is monotone and does not change sign.
\end{lemma}

\begin{proof}We only consider the case $x_0>0$. Let $\xi\in\cX$ be an admissible strategy for which $X^\xi$ is not monotone or changes sign. Define the stopping time $$\tau:=\inf\Big\{t\ge0\,\Big|\,\int_0^t\xi_s^-\,L(ds)=-x_0\Big\},
$$
where $\xi^-$ denotes the negative part of $\xi$. 
Then 
$\eta_t:=-\xi^-_t\Ind{\{ t\le\tau\}}$ 
is the desired admissible strategy with strictly lower cost than $\xi$.
\end{proof}

\begin{lemma} \label{lemma-q}
With the notations of Theorem~\ref{main thm},
$$u(T-t,S_t)\lra \varrho(S_T)\qquad\text{ $P_{z}$-a.s.~as $t \uparrow T$.}
$$
\end{lemma}

\begin{proof}By Proposition~\ref{superprocess prop}, $u$ is the unique nonnegative solution of the nonlinear integral equation
 \begin{align}\label{u int eq}
 u(t,z)=E_z\Big[\,\varrho(S_t)+\int_0^t\phi(S_s)\,ds\,\Big]-E_z\Big[\,\frac1\beta\int_0^tu(t-s,S_s)^{1+\beta}\,L(ds)\,\Big].
 \end{align}
 In particular, we get 
 \begin{align}\label{u bound}
 0\le u(t,z)\le C_\varrho+TC_\phi=: K_T  \qquad\text{for $z\in\bR$ and $t\le T$,}
 \end{align}
 if $C_\varrho$ and $C_\Phi$ are upper bounds for $\varrho$ and $\phi$.
  From the Markov property of $S$ we get that
   \begin{align}\label{u cond exp}
 u(T-t,S_t)=E_z\Big[\,\varrho(S_T)+\int_t^T\phi(S_s)\,ds\,\Big|\,\cF_t\Big]-E_z\Big[\,\frac1\beta\int_t^Tu(T-s,S_s)^{1+\beta}\,L(ds)\,\Big|\,\cF_t\Big].
\end{align}
Thus, dominated  and martingale convergence yield that $P_z$-a.s.
$$
\lim_{t \uparrow T} E_z\Big[\,\varrho(S_T)+\int_t^T\phi(S_s)\,ds\,\Big|\,\cF_t\Big]=\varrho(S_T).
$$
Moreover,~\eqref{local time identity} yields that
\begin{align*}
0&\le E_z\Big[\,\int_t^Tu(T-s,S_s)^{1+\beta}\,L(ds)\,\Big|\,\cF_t\Big]\le K_T^{1+\beta}E_{S_t}[\,L_{T-t}\,]=K_T^{1+\beta}\int_0^{T-t}p_s(S_t,c)\,ds\lra0,
\end{align*}
 as $t\uparrow T$, $P_z$-a.s. This proves the result. 
\end{proof}

\begin{proof}[Proof of Theorem~\ref{main thm}] The proof is based on some ideas from  the proof of Theorem 2.8 in~\cite{SchiedFuel}. For $p\geq2 $ we introduce the  function
\begin{align}\label{phi-funct}
\Phi_p(x,y):=x^p-py^{p-1}x+(p-1)y^p, \ \  x,y\geq 0.
\end{align}
From Young's inequality we get
\begin{align*}
xy^{p-1}\leq \frac{1}{p}x^{p}+\frac{1-p}{p}y^{p},
\end{align*}
and hence that
\begin{align} \label{phi0}
\Phi_p(x,y) \ge 0  \textrm{ with equality if and only if } x=y.
\end{align}

Now let $x_0\in\bR$ and $\xi\in\cX$ be given. By Lemma~\ref{lemma-mon-strat} we may assume that $x_0\ge0$, $\xi_t\le0$, and $X^\xi_t\ge0$  for all $t$. We write $X_t:=X^\xi_t=x_0+\int_0^t\xi_s\,ds$ and define
\begin{align}\label{Ct def eq}
C_t:=\int_0^t|\xi_s|^p\,L(ds)+\int_0^t\phi(S_s) X_s^p\,ds+X_t^pu(T-t,S_t).
\end{align}
The first two terms on the right-hand side represent the cost resulting from using the strategy $\xi$ throughout the time interval $[0,t]$. The rightmost term represents  the asserted minimal cost incurred over the time interval $[t,T]$ when starting at time $t$ with the remainder $X_t$. By Lemma~\ref{lemma-q} we have $P_z$ as $t\uparrow T$
$$C_t\lra C_T:=\int_0^T|\xi_s|^p\,L(ds)+\int_0^T\phi(S_s)X_s^p\,ds+X_T^p\varrho(S_T).
$$
Note that the expectation of the right-hand side is equal to the cost functional~\eqref{cost functional} of $\xi$. 

Let 
$$M_t:=u(T-t,S_t)+\int_0^t\phi(S_s)\,ds-\frac1\beta\int_0^tu(T-s,S_s)^{1+\beta}\,L(ds).
$$
It follows from~\eqref{u cond exp} and  Proposition~\ref{continuous strategies lemma} (whose proof does not depend on Theorem~\ref{main thm}) that $M$ is a continuous martingale. In particular, $u(T-t,S_t)$ is a continuous semimartingale. Applying It\=o's formula to~\eqref{Ct def eq} therefore yields that
\begin{align*} 
dC_t &=|\xi_t|^p\,L(dt)+X_t^p\phi(S_t)\,dt+pX_t^{p-1}u(T-t,S_t)\,dX_t+X_t^{p}\,du(T-t,S_t)  \\
&= \Big(|\xi_t|^p+pX_t^{p-1}\xi_tu(T-t,S_t)+\frac1\beta X_t^{p}u(T-t,S_t)^{1+\beta}\Big)\,L (dt)+X_t^{p}\,dM_t   \\
&=\Phi_p(|\xi_t|,X_tu(T-t,S_t)^{1/(1-p)})\,L  (dt)+X_t^{p}\,dM_t,
\end{align*} 
where we have used~\eqref{phi-funct} and the fact that $\beta+1=p/(p-1)$.  Thus, using that $X_t$ is nonnegative and nondecreasing, hence bounded,
\begin{align*} 
E_{z}[\,C_T-C_0\,]&=E_{z}\bigg[\,\int_{0}^{T}|\xi_s |^{p}\,L(ds)+\int_{0}^{T}X_s^p\phi(S_s)\,ds +\varrho(S_T)X_T^p\,\bigg] -C_0 \\
&=E_{z}\bigg[\,\int_{0}^{T}\Phi_p(|\xi_t|,X_tu(T-t,S_t)^{1/(1-p)})\,L  (dt)\,\bigg].
\end{align*}
It follows from~\eqref{phi0} that the latter expectation is nonnegative and that it vanishes if and only if 
$|\xi_t|=X_tu(T-t,S_t)^{1/(1-p)}$ holds $dL_t\otimes P_z$-almost everywhere, which in our present setting is equivalent to the equation
\begin{align}\label{optimal strategy eqn}
dX_t=- X_tu(T-t,S_t)^\beta\,dL_t. 
\end{align}
Clearly, $X^{\xi^*}$ as defined in the assertion does solve~\eqref{optimal strategy eqn}, and it follows from   Groh's generalized Gronwall inequality~\cite{Groh} that solutions to~\eqref{optimal strategy eqn} with given initial value are $P_z$-a.s.~unique. Therefore, $\xi^*$ is the  $dL_t\otimes P_z$-a.e.~unique optimal strategy in $\cX$. Finally, note that we also have $
C_0=x_0^pu(0,S_0)$, and so the theorem is proved.\end{proof}

\begin{lemma} \label{lemma-RGBM} Let $S$ be a geometric Brownian motion reflecting at $c>0$, and denote by $L$ its local time in $c$. Then there exists an arithmetic Brownian motion $\wt S$ reflecting at $\wt c:=\log c$ such that 
\begin{align} \label{s-exp} 
S_t=\exp(\wt S_t),  \qquad\text{for all $ t\geq 0$, $P_z$-a.s.} 
\end{align} 
If moreover $\wt L$ denotes the local time of $\wt S$ in $\wt c$, then 
\begin{align} \label{local-time-con} 
L_t=c\wt L_t,  \qquad\text{for all $ t\geq 0$, $P_z$-a.s.}
\end{align} 
\end{lemma}

\begin{proof}Let $S$ be as in~\eqref{reflecting GBM} and define $\wt U_t:=\frac1c U_t$. Since $dU_t$ is concentrated on $\{t\,|\,S_t=c\}$ we have $dU_t=S_t\,d\wt U_t$, and so~\eqref{reflecting GBM} yields that 
$$S_t=\exp\bigg(\log S_0+\sigma B_t+\Big(b-\frac{\sigma^2}2\Big)t+\wt U_t\bigg). 
$$
When letting 
$$\wt S_t:= \log S_0+\sigma B_t+\Big(b-\frac{\sigma^2}2\Big)t+\wt U_t
$$
we clearly have that $\wt S_t\ge\wt c$. Moreover, by construction,  $d\wt U_t$ is concentrated on $\{t\,|\,S_t=c\}=\{t\,|\,\wt S_t=\wt c\}$. Therefore, $\wt S$ is an arithmetic  Brownian motion reflecting at $\wt c$.  Finally, the formula $L_t=c\wt L_t$ follows from Exercise 1.24 in Chapter 6.1 of~\cite{RevuzYor}.
\end{proof}

\begin{proof}[Proof of Lemma~\ref{Skorohod lemma}]  We prove Lemma~\ref{Skorohod lemma} for each of the two cases separately. \medskip \\
{\bf Case 1:} $ S$ is a reflecting arithmetic Brownian motion with drift.
First, the translation invariance of Brownian motion implies that we may assume without loss of generality that $c=0$.  Moreover, we can take $\sigma=1$. In a second step,  we  use~\cite[Theorem 3.1]{Peskir} to obtain the following explicit representation of the joint law of the two processes $S$ and $L$. Let $B$ be a standard Brownian motion, and define 
$$B^{-b}_t:=B_t-bt\qquad\text{and}\qquad M^{-b}_t:=\max_{s\le t} B^{-b}_s.
$$
Then, recalling that we assume $c=0$,~\cite[Theorem 3.1]{Peskir} states that
$$
\big(S_0\vee M^{-b}-B^{-b},S_0\vee M^{-b}-S_0\big)\stackrel{\text{law}}{=}(S,L).
$$
Since the law of $(B^{-b}_t)_{0\le t\le T}$ is equivalent to the law of $(B_t)_{0\le t\le T}$, it follows that we can assume without loss of generality that $b=0$ as long as we are interested in obtaining almost-sure statements over a finite time horizon. 
Under  the assumption $b=0$, the joint law of $S$ and $L$ can also be represented as the law of $(|S_0+B|,L^{-S_0})$, where $L^{-S_0}$ is the local time of $B$ in $-S_0$.  From here, the result now follows from the corresponding and well-known result for the joint law of $B$ and $L^{-S_0}$; see, e.g., Proposition~3.1 and Th\'eor\`eme~3.2 in~\cite{LeGallDEA}. \medskip 

\noindent{\bf Case 2:} $S$ is a reflecting geometric Brownian motion. Let $\wt S$ be as in Lemma~\ref{lemma-RGBM}. Since the exponent function is locally Lipschitz continuous and the continuous trajectory $t\mapsto \wt S_t(\om)$ is bounded on $[0,T]$ for each $\om$, it follows from (\ref{s-exp}) and the result for Case 1 that 
\begin{align*} 
\sup_{t\le T}\big|S^{(n)}_{\lfloor 2^{2n}t\rfloor}-S_{t}\big|\lra0. 
\end{align*}
Note that 
\begin{align}  \label{disc-l-t-conv} 
\ell^{(n)}_k:=  c \tilde \ell^{(n)}_k, 
\end{align}
where $\tilde \ell^{(n)}_k$ is the discretized local time at $\tilde c=\log c$ of the arithmetic Brownian motion $\wt S$, which was introduced in (\ref{bm-dics-lt}).
From (\ref{local-time-con}), (\ref{disc-l-t-conv}) and the result for Case 1 we get that   
\begin{align*} 
 \sup_{t\le T}\big|\ell^{(n)}_{\lfloor 2^{2n}t\rfloor}-L_{t}\big|\lra0
 \end{align*} 
also holds in Case 2.\end{proof}

\begin{remark}It was shown in the proof of in~\cite[Proposition~3.1]{LeGallDEA} that, in the case in which $S$ is a standard Brownian motion, the stopping times $\tau^{(n)}_k$ have the property that, $P_z$-a.s., $\tau^{(n)}_{\lfloor 2^{2n}t\rfloor}\lra t$ locally uniformly in $t$. By applying the arguments from the proof of  Lemma~\ref{Skorohod lemma}, we thus obtain that for all $T>0$
\begin{align}\label{tau convergence}
\sup_{t\le T}\big|\tau^{(n)}_{\lfloor 2^{2n}t\rfloor}-t\big|\lra0\qquad\text{$P_z$-a.s.}, 
\end{align}
 for the case where $S$ is a reflecting arithmetic Brownian motion with drift. In Case 2, the stopping times $\tau_k^{(n)}$ have precisely the form~\eqref{taunkfor ABM} when $S$ is replaced by $\wt S$, where $\wt S$ is as in  Lemma~\ref{lemma-RGBM}. Therefore~\eqref{tau convergence} also holds in Case 2.\end{remark}

\begin{proof}[Proof of Proposition~\ref{continuous strategies lemma}]   
We prove Proposition~\ref{continuous strategies lemma} for each of the two cases separately. \medskip \\
{\bf Case 1:} $ S$ is a reflecting arithmetic Brownian motion with drift.
 By our formula~\eqref{optimal strategy formula}
 it is sufficient to show that the function $u(t,z)$ defined in~\eqref{u def eq} is jointly continuous in $t$ and $z$. Recall the integral equation~\eqref{u int eq} satisfied by $u$.  Clearly, the expression $h(t,z):= E_z[\,\varrho(S_t)+\int_0^t\phi(S_s)\,ds\,]$ is jointly continuous in $t$ and $z$ due to the strong Feller property of $S$ and standard arguments. 
 To show the continuity of the second expectation on the right-hand side  of~\eqref{u int eq}, we  assume $c=0$  and $\sigma=1$ as in the proof of Lemma~\ref{Skorohod lemma}. By the arguments used in the proof of that lemma, 
  \begin{align}\label{u loc tim term eq}
E_z\Big[\,\int_0^tu(t-s,S_s)^{1+\beta}\,L(ds)\,\Big]=E\bigg[\,e^{bB_T-\frac12b^2T}\int_0^tu(t-s,0)^{1+\beta}\,L^{-z}(ds)\,\bigg],
 \end{align}
 where $t\le T$, $B$ is a standard Brownian motion  under $P$, and $L^{-z}$  is the local time of $B$ in $-z$. 
 We now investigate the continuity of $u(t,z)$ as a function of $t$ when $z$ is fixed.   For $\delta>0$ and $T\ge t+\delta$ we get from~\eqref{u int eq},~\eqref{u loc tim term eq}, and our uniform bound~\eqref{u bound}
 on $u$ that
 \begin{align*}
 \big|u(t+\delta,z)-u(t,z)\big|&\le  \big|h(t+\delta,z)-h(t,z)\big|+\frac{K_T^{1+\beta}}\beta E\Big[\,e^{bB_T-\frac12b^2T}(L^{-z}_{t+\delta}-L_t^{-z})\,\Big]\\&\qquad+\frac{K_T^\beta}\beta E\bigg[\,e^{bB_T-\frac12b^2T}\int_0^t\big|u(t+\delta-s,0)-u(t-s,0)\big|\,L^{-z}(ds)\,\bigg]
 \end{align*}
When letting
 $$g_\delta(t,z):=\big|h(t+\delta,z)-h(t,z)\big|+\frac{K_T^{1+\beta}}\beta E\Big[\,e^{bB_T-\frac12b^2T}(L^{-z}_{t+\delta}-L_t^{-z})\,\Big],
 $$
 we thus get that $w_\delta(t,z):= \big|u(t+\delta,z)-u(t,z)\big|$ satisfies
 $$w_\delta(t,z)\le g_\delta(t,z)+\frac{K_T^\beta}\beta E\bigg[\,e^{bB_T-\frac12b^2T}\int_0^tw_\delta(t-s,0)\,L^{-z}(ds)\,\bigg].
 $$
 When defining the continuous measure $\mu_z$ through
 $$\mu_z([0,t])=\frac{K_T^\beta}\beta E\bigg[\,e^{bB_T-\frac12b^2T}L_t^{-z}\,\bigg],
 $$
  Groh's generalized Gronwall inequality~\cite{Groh} (see also Theorem 5.1 in Appendix~5 of~\cite{EthierKurtz}) now yields that  
 $$w_\delta(t,0)\le \sup_{s\le T-\delta}g_\delta(s,0)e^{\mu_0([0,t])}.
  $$
  Therefore,
  \begin{align}\label{wdelta estimate}
  \sup_{t\le T-\delta,\,|z|\le r}w_\delta(t,z)\le \sup_{s\le T-\delta,\,|z|\le r}g_\delta(s,z)\Big(1+e^{\mu_0([0,T])}\sup_{z\in\bR} \mu_z([0,T])\Big).
  \end{align}
   The expression $\sup_{z} \mu_z([0,T])$ is finite, because it follows from~\cite[Lemma 1]{Csaki} that
 \begin{align}\label{local time exp moments}
 E\Big[\,\exp\Big(\lambda\sup_{z\in\bR}L^z_T\Big)\,\Big]<\infty,\qquad\text{for all $\lambda\ge0$.}
 \end{align}
 Therefore the right-hand side of~\eqref{wdelta estimate} tends to zero for $\delta\da0$, locally uniformly in $z$, and we conclude that $ u(t,z)$ is locally uniformly continuous in $t$, locally uniformly in  $z$.  
 
 The joint continuity of $(t,z)\mapsto u(t,z)$ will now follow if we show that $z\mapsto u(t,z)$ is continuous for each fixed $t$. 
But this continuity   follows via dominated convergence from~\eqref{local time exp moments},~\eqref{u int eq},~\eqref{u loc tim term eq},  our uniform bound on $u$, and  the joint continuity of $L^{-z}_t$ in $t$ and $z$.\medskip \\
{\bf Case 2:} $S$ is a reflecting geometric Brownian motion. Use the notation of Lemma~\ref{lemma-RGBM} together with (\ref{s-exp}) to get 
\begin{align*}
u(t,z)&=E_z\Big[\,\varrho(e^{\widetilde S_t})+\int_0^t\phi(e^{\widetilde S_s})\,ds -\frac{1}{\beta}\int_0^tu(t-s,e^{\widetilde S_s})^{1+\beta}\,L(ds)\,\Big].  
\end{align*}
Hence by setting $\tilde z =\log z$, $\tilde u(t,x) = u(t,e^{x})$, $\tilde \phi(x) =   \phi(e^x)$ and $\tilde \varrho(x) =   \varrho(e^x)$ and using (\ref{local-time-con}) we get 
\begin{align*}
\tilde u(t, \tilde z)&=E_{\tilde z}\Big[\, \tilde \varrho(\widetilde S_t)+\int_0^t \tilde \phi(\widetilde S_s)\,ds -\frac{c}{\beta}\int_0^t\tilde u(t-s,\widetilde S_s)^{1+\beta}\,\widetilde L(ds)\,\Big],  
\end{align*}
and the joint continuity of $u$ follows immediately from Case 1.  
\end{proof}

\begin{proof}[Proof of Proposition~\ref{convergence prop}] It follows from Lemma~\ref{Skorohod lemma} that the finite and nonnegative Radon measures associated with the right-continuous nondecreasing functions $t\mapsto \ell^{(n)}_{\lfloor 2^{2n}t\rfloor}$ converge vaguely and $P_z$-a.s.~to the continuous measure $dL_t$. Therefore, and by the continuity of $t\mapsto\xi_t$,
$$\int_0^{t}\xi_s\,d \ell^{(n)}_{\lfloor 2^{2n}s\rfloor}\lra \int_0^t\xi_s\,dL_s\qquad \text{$P_z$-a.s.}
$$
Moreover, it follows from~\eqref{tau convergence} that 
$$\sup_{s\le t}\big|\xi_s-\xi_{\tau^{(n)}_{\lfloor 2^{2n}s\rfloor}}\big|\lra0\qquad\text{$P_z$-a.s.}
$$
Therefore, for each fixed $t$, 
$$\sum_{k=0}^{\lfloor 2^{2n}t\rfloor}\xi^{(n)}_k(\ell^{(n)}_k-\ell^{(n)}_{k-1})=\int_0^{t}\xi_{\tau^{(n)}_{\lfloor 2^{2n}s\rfloor}}\,d \ell^{(n)}_{\lfloor 2^{2n}s\rfloor}\lra \int_0^t\xi_s\,dL_t\qquad \text{$P_z$-a.s.}
$$
This shows that, $P_z$-a.s., $X^{\xi,(n)}_{\lfloor 2^{2n}t\rfloor}\to X^\xi_t$ for all rational $t$, and thus, by continuity and Lemma~\ref{Skorohod lemma}, for all $t$. The convergence statement for the accumulated transaction costs  follows by applying the same argument to $|\xi_s|^p$ instead of $\xi_s$. 
\end{proof}

\noindent{\bf Acknowledgement.} The authors thank the Institute for Pure and Applied Mathematics at UCLA for the invitations to  the workshop on \emph{The Mathematics of High Frequency Financial Markets}, during which the research on this paper was completed. The authors furthermore thank Leo Speiser for helping with Figure~\ref{EUR_CHF Fig}. A.S.~gratefully acknowledges financial support by Deutsche Forschungsgemeinschaft through  Research Grant SCHI/3-2.

\end{document}